\documentclass[journal,twoside,web]{ieeecolor}
\usepackage{generic}

\showboxdepth=\maxdimen
\showboxbreadth=\maxdimen

\let\labelindent\relax
\usepackage[english]{babel}
\usepackage{amsmath,mathtools,amssymb,amsthm}
\usepackage{cleveref,cite}
\usepackage{enumitem}
\usepackage{etoolbox}
  
\usepackage{microtype}
\usepackage{graphicx}
\usepackage{textcomp}
\usepackage{array}
\newcommand{\CS}[1]{\mathcal{C\!S}\lc#1\rc}
\newcommand{\hs}[1]{\hspace{#1 cm}}
\newcommand{\change}[1]{{#1}}

\makeatletter 
\pretocmd\@bibitem{\color{black}\csname keycolor#1\endcsname}{}{\fail}
\newcommand\citecolor[1]{\@namedef{keycolor#1}{\color{black}}}
\makeatother

\citecolor{chapman2015state}
\citecolor{mousavi2021strong}
\citecolor{mousavi2017structural}
\citecolor{commault2019functional}
\citecolor{jadbabaie2018deterministic}
\citecolor{nagahara2011sparse}
\citecolor{li2016sparse}
\citecolor{degroot1974reaching}
\citecolor{golub2012homophily}
\citecolor{nabi2014social}
\citecolor{shu2019studying}
\citecolor{cremonini2017controllability}

\crefname{app}{Appendix}{Appendices}
\crefname{cor}{Corollary}{Corollaries}
\crefname{prop}{Proposition}{Propositions}
\crefname{lemma}{Lemma}{Lemmas}
\crefname{defn}{Definition}{Definitions}
\crefname{conj}{Conjecture}{Conjectures}
\crefname{exam}{Example}{Examples}

\newtheorem{theorem}{Theorem}
\newtheorem{cor}{Corollary}
\newtheorem{lemma}{Lemma}

\newtheorem{exam}{Example}

\newcommand{\bs}{\boldsymbol}
\newcommand{\bb}{\mathbb}
\newcommand{\mcal}{\mathcal}

\newcommand{\eye}{\bs{I}}
\newcommand{\zero}{\bs{0}}

\newcommand{\lb}{\left(}
\newcommand{\rb}{\right)}
\newcommand{\ls}{\left[}
\newcommand{\rs}{\right]}
\newcommand{\lc}{\left\{}
\newcommand{\rc}{\right\}}

\newcommand{\lv}{\left\vert}
\newcommand{\rv}{\right\vert}

\newcommand{\LRV}[1]{{\left\vert\kern-0.25ex\left\vert\kern-0.25ex\left\vert #1 \right\vert\kern-0.25ex\right\vert\kern-0.25ex\right\vert}}

\newcommand{\nth}{^\mathsf{th}}
\newcommand{\tran}{^{\mathsf{T}}}

\newcommand{\rank}[1]{\mathsf{Rank}\lc#1\rc}

\newcommand{\matA}{\bs{A}}
\newcommand{\matB}{\bs{B}}
\newcommand{\matC}{\bs{C}}

\newcommand{\matJ}{\bs{J}}

\newcommand{\matN}{\bs{N}}

\newcommand{\matP}{\bs{P}}
\newcommand{\matQ}{\bs{Q}}
\newcommand{\matR}{\bs{R}}

\newcommand{\matW}{\bs{W}}

\newcommand{\bbC}{\bb{C}}

\newcommand{\bbR}{\bb{R}}

\newcommand{\calA}{\mcal{A}}

\newcommand{\calS}{\mcal{S}}

\newcommand{\calU}{\mcal{U}}
\newcommand{\calV}{\mcal{V}}

\newcommand{\vecu}{\bs{u}}
\newcommand{\vecv}{\bs{v}}

\newcommand{\vecx}{\bs{x}}
\newcommand{\vecy}{\bs{y}}
\newcommand{\vecz}{\bs{z}}

\def\BibTeX{{\rm B\kern-.05em{\sc i\kern-.025em b}\kern-.08em
    T\kern-.1667em\lower.7ex\hbox{E}\kern-.125emX}}
\begin{document}
\title{Output Controllability of a Linear Dynamical System with Sparse Controls}
\author{Geethu Joseph
\thanks{The author is with the Department  of Electrical Engineering and Computer Science at the Syracuse University, NY 13244, USA, Email: gjoseph@syr.edu.}}
\maketitle
\begin{abstract}
In this paper, we study the conditions to be satisfied by a discrete-time linear system to ensure output controllability using sparse control inputs. A set of necessary and sufficient  conditions can be directly obtained by extending the Kalman rank test for output controllability. However, the verification of these conditions is  computationally heavy due to their combinatorial nature. Therefore, we derive non-combinatorial conditions for output sparse controllability that can be verified with polynomial time complexity. Our results also provide bounds on the minimum sparsity level required to ensure output controllability of the system. This additional insight is useful for designing sparse control input that drives the system to any desired output.
\end{abstract}
\section{Introduction}
With the widespread acceptance and use of networked control systems, various new challenging theoretical issues have emerged in control theory. One such 
problem is the analysis of a network system with \emph{sparse control inputs}. In particular, the controllability of systems under the sparsity constraints on the input is a relatively new concept~\cite{joseph2020controllability,joseph2020pcontrollability}. This paper characterizes \emph{output controllability} of a linear system with sparse control, i.e., the input applied at every time instant has a few nonzero entries compared to its dimension.

\change{\subsection{Practical Context and Examples}\label{sec:motivation}
Constraining the inputs to be sparse is often necessary to select a small subset of the available sensors or actuators at each time instant, due to bandwidth, energy, or physical constraints. The sparse control inputs arise in several areas like multiagent systems~\cite{caponigro2015sparse}, optimal actuator placement~\cite{stadler2009elliptic,chanekar2017optimal}, nodes selection~\cite{ikeda2018sparse,ikeda2018sparsity}, opinion dynamics~\cite{joseph2020ocontrollability}, environmental monitoring systems~\cite{joseph2020controllability,sriram2020control}, and robotics~\cite{vossen2006l1}, to name a few. 

In the following, we discuss two examples of the systems wherein the input is sparse and its support is time-varying.

\subsubsection{Networked Control Systems} The networked control systems, where the controller and plant communicate over a network, are often constrained by energy and bandwidth. In energy-constrained networks, energy-aware scheduling of actuators can help to extend the battery life of the nodes~\cite{jadbabaie2018deterministic}. In this case, choosing the same support for a longer time drains the battery of the selected set of nodes. Therefore, using a different set of nodes at each time instant can result in a longer network lifetime. Further, in networked control systems, the control inputs are required to meet the bandwidth constraints imposed by the links over which they are exchanged~\cite{nagahara2011sparse,li2016sparse}. The sparse vectors admit compressed representations, and consequently, using sparse inputs helps to reduce the bandwidth requirements~\cite{Donoho_Compressed_2006,
Candes_Robust_2006, Baraniuk_Compressive_2007}. Also, restricting the control inputs to a fixed support may severely limit the set of admissible inputs to the system. On the other hand, using different supports provides much greater flexibility without significantly increasing the communication requirements. Thus, this approach combines the benefits of the other two methods.

\subsubsection{Social Networks}\label{sec:network}A social network is often modeled using a graph whose vertices represent the individuals in the network, and the edges represent the social connection between the individuals. A popular model for the evolution of network opinion is the DeGroot model that uses a linear dynamical system~\cite{degroot1974reaching,golub2012homophily}. Here, the system state is denoted by a vector containing the opinion of each individual in the network, and the transition matrix is the adjacency matrix of the graph. Further, it is assumed that an external agent such as an election candidate, a paid bloggers, and a marketing agent, desires to drive the network opinion to a particular state by influencing only a few people on the network~\cite{nabi2014social,shu2019studying,cremonini2017controllability}. For example, consider an election candidate visiting the voters as part of a political campaign~\cite{joseph2020ocontrollability}. At each time instant, the candidate can only visit and influence a small group of people. Also, for better campaigning, the candidate does not visit the same set of people at each time. As a result, the support of the sparse control input (due to the candidate) also varies with time. Moreover, the goal of the candidate is not to influence all the voters, but it is enough if the candidate can influence the majority of the voters. Hence, the candidate designs the campaign strategy such that a subset of the network opinion can be driven to the desired value. This problem can be solved via the analysis of the output controllability of the network opinion. 
}

\subsection{Related Literature}\label{sec:related}
Before we present our model and results, we provide a brief review of the existing literature on sparse control. 
\subsubsection{Structural Controllability of Networks} The characterization of controllability of networks by using a few nodes is a well-studied problem\change{~\cite{liu2011controllability,egerstedt2012interacting,chapman2013strong,chapman2014controllability,
pequito2015framework,chapman2015state,mousavi2017structural,mousavi2021strong,
commault2019functional}}. However, these papers focused on \change{structural controllability or strongly structural controllability}. On the contrary, we deal with output controllability by assuming the knowledge of the system matrices. Also, in such problems, sparsity refers to the number of driver nodes and not the number of nonzero entries in the control input.

\subsubsection{Minimal Controllability Problem}
The minimal controllability refers to the problem of selecting a small set of input variables so that the system is controllable using the selected set~\cite{liu2017minimal,tzoumas2015minimal,olshevsky2014minimal}. This model is similar to ours except that the support of the control inputs does not change with time. Here, the support refers to the indices of the nonzero entries of the input. However, the time-varying support model is more flexible and offers better control over the system while incurring a similar cost (in terms of energy and bandwidth) as that of the \change{time-invariant} support model~\cite{joseph2020controllability}. Therefore, we analyze the controllability of a linear system with inputs having supports that change with time.

\subsubsection{Time-varying Actuator Scheduling}
The time-varying actuator scheduling deals with the control of linear systems using sparse inputs with a time-varying support model~\cite{dorfler2014sparsity,ito2014optimal,chatterjee2016characterization,kalise2017infinite,rao2017dynamic,zhou2017minimal,
ikeda2016maximum,ikeda2018sparse,ikeda2018sparsity}. However, previous studies on this problem mainly focused on the design of sparse control inputs and the optimal actuator scheduling (choosing the support of the control inputs at every time instant) problems. Such problems were formulated as optimization problems with $\ell_0$-norm constraint on the input. The $\ell_0$-norm-based problems are  NP-hard, and thus, they were solved using $\ell_p$-norm-based relaxation ($0<p\leq 1$) or greedy algorithms. While these studies attempted to devise approximation algorithms to design the control inputs, our focus in this paper is to gain new fundamental insights into the conditions for output controllability of a system using sparse inputs that follow the time-varying support model. 

\subsubsection{Sparse Controllability}
Sparse-controllability defined in \cite{joseph2020controllability} refers to the controllability of a linear system when the inputs are sparse, and their supports are time-varying. In  \cite{joseph2020controllability}, the authors derived the necessary and sufficient conditions for sparse controllability that are non-combinatorial. In particular, they established that any controllable system is sparse controllable if and only if the sparsity level exceeds the nullity of the \change{state matrix}.\footnote{The precise statements of the results are presented in \Cref{sec:compare}.} This work was also extended to controllability using nonnegative sparse control inputs~\cite{joseph2020pcontrollability}. However, in \cite{joseph2020controllability}, the authors only dealt with state controllability. A similar algebraic characterization of output sparse controllability is not straightforward. This is because the results on sparse controllability in \cite{joseph2020controllability} are based on the Popov–Belevitch–Hautus (PBH) test~\cite{hautus1970stabilization} for controllability. However, an analogous PBH test for output sparse controllability is not available in the literature. Consequently, the proof technique used in \cite{joseph2020controllability} is not applicable for output sparse controllability.

In a nutshell, in this paper, we derive the conditions for output sparse controllability of a linear system using the fundamental tools from linear algebra and matrix theory. 

\subsection{Our Contributions}
 We  present a discrete-time \change{linear time-invariant dynamical system} with sparse control inputs and time-varying support model in \Cref{sec:system}. We then show that the direct extension of the Kalman type rank test for output sparse controllability leads to a verification procedure with exponential time complexity. In \Cref{sec:necc_suff}, we show that any linear system that is output controllable is also output sparse controllable if and only if the sparsity level exceeds a certain bound which we present in \Cref{thm:necc_suff}. Hence, our result also provides the minimum sparsity level that ensures output controllability. In addition, we present several implications and insights from our result and compare it with the existing results on controllability and sparsity in \Cref{sec:intution,sec:sim,sec:complex,sec:add,sec:compare}. Finally, we discuss the design of sparse control inputs that drive the system output to a desired value in \Cref{sec:design}.

\null

\noindent\emph{Notation:} In the sequel, boldface lowercase letters denote vectors, boldface uppercase letters denote matrices, and calligraphic letters denote sets. \change{The $i\nth$ column of the matrix $\matA$ is denoted by $\matA_i$ while the submatrix of $\matA$ formed by the columns indexed by the set $\calA$ is denoted by $\matA_{\calA}$. The symbols $(\cdot)\tran$, $\rank{\cdot}$, $(\cdot)^{-1}$, $(\cdot)^\dagger$, and $\CS {\cdot}$ denote the transpose, rank,  inverse, pseudo-inverse, and column space of a matrix, respectively. Also, the cardinality of a set is denoted using $\lv\cdot\rv$, and the ceiling function is denoted using $\lceil\cdot\rceil$. Further, the notation $\eye$ and $\zero$ represent the identity matrix and the zero matrix (or vector), respectively. Finally, we use $\bbR$ to denote the set of real numbers and $\bbC$ for the set of complex numbers.
 }

\section{Output Sparse Controllability}\label{sec:system}
We consider the discrete-time linear dynamical system described by the triple $\lb\matA,\matB,\matC\rb$ in which the state and output evolve as follows:
\begin{equation}
\vecx_k = \matA\vecx_{k-1}+\matB\vecu_{k}\text{ and } \vecy_k=\matC\vecx_k.\label{eq:sys}
\end{equation}
Here, $\vecx_k\in\bbR^{N}$ denotes the state vector, $\vecu_k\in\bbR^{m}$ denotes the control input vector, and $\vecy_k\in\bbR^n$  denotes the output vector at time $k$. Also, $\matA$, $\matB$, and $\matC$ are the \change{state matrix},  input matrix and output matrix of the system, respectively. We assume that the control vectors are \change{constrained to be} $s$-sparse, i.e., at most $s$ entries of $\vecu_k$ are nonzeros, for all values of $k$. Under this sparsity constraint on the input, we revisit the classical output controllability problem. To be specific, our goal is to check if it is possible to drive the output to any final state $\vecy_f\in\bbR^n$, starting from any initial state $\vecx_0\in\bbR^N$,  using $s$-sparse control inputs within a finite time. This notion of controllability is referred to as \emph{output $s$-sparse controllability}, henceforth.

Using  \eqref{eq:sys}, the output at any time $K>0$ is
\begin{equation}\label{eq:sys_concat}
\vecy_{K} = \matC\sum_{k=1}^{K}\matA^{K-k}\matB\vecu_k+\matC\matA^K\vecx_0.
\end{equation}
So the system is output $s$-sparse controllable if and only if there exists an integer $0<K<\infty$ such that 
\begin{equation}\label{eq:combinatorial_union}
\bigcup_{\substack{\{\calS_k\subseteq\{1,2,\ldots,m\}:\\~\lv\calS_k\rv\leq s, 1\leq k\leq K\}}} \hs{-.52}\CS{ \!\matC\!\begin{bmatrix}\matA^{K-1} \matB_{\calS_1} & \matA^{K-2} \matB_{\calS_2} \ldots  \matB_{\calS_K}\end{bmatrix}\!}\!=\!\bbR^n\!.
\end{equation}
\change{However, a vector space over an infinite field cannot be a finite union of proper subspaces~\cite{clark2012covering}. Then, from \eqref{eq:combinatorial_union}, output $s$-sparse controllability holds only if there exist an integer $N < K<\infty$ and index sets $\lc \calS_i,\lv\calS_i\rv\leq s\rc_{i=1}^K$ such that 
\begin{equation}\label{eq:combinatorial}
\CS{\begin{bmatrix}
\matC\matA^{K-1}\matB_{\calS_1} & 
\matC\matA^{K-2}\matB_{\calS_2}&\ldots& 
\matC\matB_{\calS_K}
\end{bmatrix}}=\bbR^n.
\end{equation}}
The direct evaluation of the condition \eqref{eq:combinatorial} requires  computation of the column spaces of $\binom{N}{s}^K$ matrices of size $n\times Ks$. Thus, the verification of the condition is computationally expensive. Motivated by this, in the next section, we present some non-combinatorial conditions that help to test output sparse controllability. 

\section{Necessary and Sufficient Conditions}\label{sec:necc_suff}
The results of this section are based the controllability matrix $\matW$ and a new metric $R_i$ as defined below: 
\begin{align}\label{eq:W_defn}
\matW &\triangleq \begin{bmatrix}
\matA^{N-1}\matB& \matA^{N-2}\matB & \ldots &\matB
\end{bmatrix}\in\bbR^{N\times Nm}\\
\change{R_i}&\change{\triangleq \rank{\matC\matA^i\matW} -\rank{\matC\matA^{i+1}\matW}},\label{eq:Rdefn}
\end{align}
where $i\geq0$ is an integer. The main result of this section is as follows:

\begin{theorem}\label{thm:necc_suff}
Consider the discrete-time linear dynamical system $\lb\matA,\matB,\matC\rb$ defined in  \eqref{eq:sys}  whose controllability matrix $\matW$ is given by \eqref{eq:W_defn}. Then, for any integer $0<s\leq m$, \change{a set of necessary conditions} for output $s$-sparse  controllability  are
\begin{equation}\label{eq:ness}
\rank{\matC\matW} = n \text{ and }
 \max_{0\leq i\leq N-1}\frac{\sum_{j=0}^i R_j}{i+1}\leq s,
\end{equation}
and \change{a set of sufficient conditions} are
\begin{equation}
\rank{\matC\matW} = n \text{ and }
 \min\lc m, \max_{0\leq i\leq N-1}  R_i\rc\leq s.\label{eq:suff}
\end{equation}
Here, $R_i$ is as defined in \eqref{eq:Rdefn}.
\end{theorem}
\begin{proof}
See \Cref{app:necc_suff}.
\end{proof}

In the following subsections, we discuss the geometric intuition and insights from \Cref{thm:necc_suff}.

\change{
\subsection{Geometric Intuition} \label{sec:intution}
The rank condition in \eqref{eq:ness} and \eqref{eq:suff} is straightforward from the Kalman rank test for (non-sparse) output controllability (see \Cref{thm:output}). The bounds on the sparsity in \Cref{thm:necc_suff} can be intuitively explained as follows. 

\subsubsection{Necessary Condition} From \eqref{eq:combinatorial}, the system is $s$-sparse output controllable \emph{if and only if} the last $(i+1)s$ columns of the matrix in \eqref{eq:combinatorial} span the left null space $\calU_{i}$ of the submatrix formed by its remaining columns. However, $\calU_{i}\subseteq\bbR^n$ contains the left null space of $\matC\matA^{i+1}\matW$ because $\CS{\matW}$ is the subspace of the state vectors that can be reached from $\vecx_0=\zero$. So we arrive at
\begin{equation}
(i+1)s \geq n-\rank{\matC\matA^{i+1}\matW} = \sum_{j=0}^i R_i.
\end{equation} 
The above relation leads to the bound on sparsity given in the necessary condition \eqref{eq:ness}. The bound is not sufficient because the rank condition does not necessarily ensure the spanning condition in \eqref{eq:combinatorial} (see \Cref{ex:ness_suff1} below).

\subsubsection{Sufficient Condition}  The left null space of $\matC\matA^{i+1}\matW$ contains that of $\matC\matA^{i}\matW$. As a consequence, one possible case where \eqref{eq:combinatorial} holds is when the column space of $\matC\matA^i\matB_{\calS_{K-i}}\in\bbR^{n\times s}$ spans the subspace $\calV_{i}$. Here, $\calV_{i}$ is the subspace of the left null space of $\matC\matA^{i+1}\matW$ which is orthogonal to that of $\matC\matA^i\matW$, and its dimension is $R_i$. This case leads to
\begin{equation}
s\geq \rank{\matC\matA^i\matB_{\calS_{K-i}}} \geq R_{i}.
\end{equation} 
The above relation leads to the bound on sparsity given in the sufficient condition \eqref{eq:suff}. The condition is not necessary because it considers only one possible case for \eqref{eq:combinatorial} to hold (see \Cref{ex:ness_suff2} below). 

Please refer to \Cref{app:necc_suff} for the rigorous proof. We illustrate our idea using the following examples:}
\begin{exam}\label{ex:ness_suff1}
Consider the system $\lb\matA,\matB,\matC\rb$ in  \eqref{eq:sys}  with
\begin{equation}
\matA= \begin{bmatrix}
0 &1 &0 & 0 &0\\
0 &0&  1 & 0 &0\\
0 &0 &0 &0&0\\
0 & 0 & 0 & 0 & 1\\
0 & 0 & 0 & 0&0
\end{bmatrix}\matB=\begin{bmatrix}
1   &  1\\
0   &  0\\
1   &   0\\
0   &   0\\
0 & 1
\end{bmatrix}\matC =\begin{bmatrix}
1   &  0 &    0 \\
0   &  1    & 0\\
0   &  0  &   0\\
0   &  0   &  1\\
0   &  0   &  0
\end{bmatrix}\tran.
\end{equation}
For this system, we have $\rank{\matC\matW}=3=n$,
\begin{equation}
\underset{0\leq i\leq N-1}{\max}  R_i=2 \text{ and }
 \underset{0\leq i\leq N-1}{\max}\frac{\sum_{j=0}^{i}R_j}{i+1}=1. 
\end{equation}
Therefore, when $s=1$, the system satisfies the necessary condition, but it does not satisfy the sufficient condition. Using the brute force verification of output sparse controllability (given by \eqref{eq:combinatorial}), we see that the system is not output $1-$sparse controllable. Thus, this example shows that the necessary conditions of \Cref{thm:necc_suff} are not always sufficient for output sparse controllability. 
\end{exam}

\begin{exam}\label{ex:ness_suff2}
Consider the system $\lb\matA,\matB,\matC\rb$ in  \eqref{eq:sys}  with
\begin{equation}
\matA= \begin{bmatrix}
0 & 1 & 0 & 0\\
0 & 0 & 0 & 0\\
0 & 0 & 0 & 1\\
0 & 0 & 0 & 0
\end{bmatrix}\matB=\begin{bmatrix}
1   &  1\\
1  &  0\\
0   &   0\\
0   &   1\\
\end{bmatrix}\matC =\begin{bmatrix}
1   &  0\\
0  &   0\\
0  &   1\\
0   &  0
\end{bmatrix}\tran.
\end{equation}
For this system, we have $\rank{\matC\matW}=2=n$,
\begin{equation}
\underset{0\leq i\leq N-1}{\max}  R_i=2 \text{ and }
 \underset{0\leq i\leq N-1}{\max}\frac{\sum_{j=0}^{i}R_j}{i+1}=1. 
\end{equation}
Therefore, when $s=1$, the system satisfies the necessary condition, but it does not satisfy the sufficient condition. However, the system defined by $(\matA,\matB_2,\matC)$ is output sparse controllable \change{where $\matB_2\in\bbR^4$ is the second column of $\matB$}. Thus, this example shows that the sufficient conditions of \Cref{thm:necc_suff} are not always necessary.
\end{exam}

\subsection{Simultaneous Necessity and Sufficiency}\label{sec:sim}
\change{Both necessary and sufficient conditions in \Cref{thm:necc_suff}} become identical under some mild assumptions on the system, which we present in the following corollary.
\begin{cor}\label{cor:necc_suff_1}
Consider the discrete-time linear dynamical system $\lb\matA,\matB,\matC\rb$ defined in  \eqref{eq:sys}  whose controllability matrix $\matW$ is given by \eqref{eq:W_defn}. The conditions \eqref{eq:ness} and \eqref{eq:suff} become identical if and only~if
$\max_{0\leq i\leq N-1} R_i=R_0$,
where $R_i$ is as defined in \eqref{eq:Rdefn}. Also, in this case,  \eqref{eq:ness} and \eqref{eq:suff} reduce to
\begin{equation}\label{eq:necc_condition}
\change{\rank{\matC\matW} = n \text{ and } s \geq n-\rank{\matC\matA\matW}.}
\end{equation}
\end{cor}
\begin{proof}
See \Cref{app:necc_suff_1}.
\end{proof}

We note that the assumption in \Cref{cor:necc_suff_1} is \change{satisfied by a large class of matrices}. For example, suppose that $ \rank{\matA}=\rank{\matA^2}$. In this case, the row space and the column space of $\matA^i$ are same as those of $\matA$,  for $i\geq 1$. As a result, we obtain
\begin{equation}
\change{\rank{\matC\matA\matW}=\rank{\matC\matA^i\matW}.}
\end{equation}
Consequently, from \eqref{eq:Rdefn}, we get $R_i = 0\leq R_0$ for all values of  $i\geq 1$. Hence, $\max_{0\leq i\leq N-1} R_i=R_0$, and by \Cref{cor:necc_suff_1}, the necessary and sufficient conditions of \Cref{thm:necc_suff} reduce to \eqref{eq:necc_condition}. Here, the condition $ \rank{\matA}=\rank{\matA^2}$ implies that the algebraic and geometric multiplicities of the eigenvalue 0 of $\matA$ are the same\change{~\cite[Chapter 3]{horn2012matrix}}. This condition is satisfied by the families of matrices like the diagonalizable matrices,  the matrices with rank greater than or equal to $N-1$, etc. 

\change{The above observation is particularly useful to analyze the network opinion of social networks like Facebook (see \Cref{sec:network}). Such a network is modeled using an undirected graph, and therefore, the state matrix $\matA$  of the corresponding linear dynamical system is the symmetric adjacency matrix of the graph~\cite{joseph2020ocontrollability}. Since symmetric matrices are always diagonalizable, \eqref{eq:necc_condition} gives the necessary and sufficient conditions in this case.}
\subsection{Computational Complexity}\label{sec:complex}
The computational complexity to verify all the conditions of \Cref{thm:necc_suff} depends on the complexity to compute the rank of matrices \change{$\matC\matA^i\matW$}, for $i=0,1,\dots,N$. So, unlike the verification of the combinatorial condition \eqref{eq:combinatorial}, the verification of the conditions in \Cref{thm:necc_suff} possesses polynomial time complexity (in $N$ and $n$), and the complexity is independent of $s$. Moreover, the complexity of the verification test can be further reduced by using a simpler condition which does not involve the computation of \change{$\lc R_i\rc_{i=0}^{N-1}$} as presented below:
\begin{cor}\label{cor:necc_suff_2}
Consider the discrete-time linear dynamical system $\lb\matA,\matB,\matC\rb$ in  \eqref{eq:sys}  whose controllability matrix $\matW$ is given by \eqref{eq:W_defn}. The system is output $s$-sparse  controllable for any $s>0$ if 
\begin{equation}\label{eq:necc_cor_2}
\rank{\matC\matW} = n \text{ and }s \geq \min\lc m,N-\rank{\matA}\rc.
\end{equation}
\end{cor}
\begin{proof}
See \Cref{app:necc_suff_2}.
\end{proof}
Clearly, the relaxed bound on $s$ given in \eqref{eq:necc_cor_2} is easy to calculate. So if the system satisfies the bound in \Cref{cor:necc_suff_2}, we can avoid the more computationally heavy conditions of \Cref{thm:necc_suff}. Also, from the proof of the result, we notice that \eqref{eq:necc_cor_2} in \Cref{cor:necc_suff_2} can also be replaced with a more stringent condition,
\begin{equation}
s \geq \rank{\matW}-\rank{\matA\matW},
\end{equation} 
which follows from the proof of \Cref{cor:necc_suff_2} (see \eqref{eq:cor_inter_4} in \Cref{app:necc_suff_2}). 

Further, \Cref{cor:necc_suff_2} implies that if a linear system is reversible, i.e., $\matA$ is nonsingular, then (output) controllability implies and is implied by (output) $s$-sparse controllability, for any $1\leq s\leq m$. This property also holds for sparse controllability~\cite{joseph2020controllability}.

\subsection{Additional Insights}\label{sec:add}
Some interesting observations from \Cref{thm:necc_suff} are as follows:
\change{\subsubsection{Minimum Sparsity} 
If the sparsity $s$ satisfies the sufficient bound in \eqref{eq:suff}, any output controllable system is guaranteed to be sparse output controllable. On the other hand, if the sparsity $s$ violates the necessary bound in \eqref{eq:ness}, the system is guaranteed not to be sparse output controllable. Consequently, the minimum sparsity level $s^*$ that guarantees sparse output controllability lies in the interval defined by the two lower bounds given in \Cref{thm:necc_suff}. Thus, we get}
\begin{equation}\label{eq:min_sparsity}
 \max_{0\leq i\leq N-1}\frac{\sum_{j=0}^i R_j}{i+1}\leq s^*\leq \min\lc m, \max_{0\leq i\leq N-1}  R_i\rc.
\end{equation}
\subsubsection{Bound on Necessary Sparsity} Since the sufficient conditions are more stringent than the necessary conditions, from \eqref{eq:ness} and \eqref{eq:suff}, we arrive at 
\begin{equation}\label{eq:bound_m}
m\geq\max_{0\leq i\leq N-1}\frac{\sum_{j=0}^i R_j}{i+1}.
\end{equation}
This condition holds for any output (non-sparse) controllable systems. 
\subsubsection{Time-invariant system} We note that all the systems that are output controllable using sparse inputs with the time-invariant support model are also output controllable using sparse inputs with the time-varying support model.  Hence, \eqref{eq:ness} is necessary for output controllability using sparse inputs with time-invariant support model.
\subsubsection{Non-canonical Basis} If a system is output controllable using control inputs which are $s$-sparse in the canonical basis, it is output controllable using inputs that admit $s$-sparse representations under any other basis $\Phi\in\bbR^{m\times m}$. This is because the change of basis is equivalent to replacing $\matB$ with $\matB\Phi$ which does not change the  condition in \eqref{eq:combinatorial}. Sparse controllability also possess a similar property~\cite{joseph2020controllability}.

\subsection{Comparison with Existing Results}\label{sec:compare}
In this section, we compare \Cref{thm:necc_suff} with the existing results on controllability and sparsity.
\renewcommand{\thetheorem}{\Alph{theorem}}
\setcounter{theorem}{0}

\subsubsection{Output controllability without constraints} 
The classical result for output controllability is as follows.
\begin{theorem}[{\cite{sontag2013mathematical}}]\label{thm:output}
Consider the linear dynamical system $\lb\matA,\matB,\matC\rb$ defined in  \eqref{eq:sys}  whose controllability matrix $\matW$ is given by \eqref{eq:W_defn}. The system is output controllable if and only if $\rank{\matC\matW}=n$.
\end{theorem}
 If we remove the sparsity constraint, i.e., when $s=m$, \Cref{thm:necc_suff} coincides with \Cref{thm:output}, as expected.
\subsubsection{Controllability with sparse inputs} 
The next result gives the necessary and sufficient conditions for  controllability with sparse control inputs. 
\begin{theorem}[{\cite[Theorem 1]{joseph2020controllability}}]\label{thm:sparse}
Consider the linear dynamical system $\lb\matA,\matB,\matC\rb$ defined in  \eqref{eq:sys}  whose controllability matrix $\matW$ is given by \eqref{eq:W_defn}. The system is controllable using $s$-sparse inputs if and only if the following conditions hold:
\begin{align}\label{eq:sparse_con} 
\rank{\begin{bmatrix}
\lambda\eye-\matA & \matB
\end{bmatrix}} &= N \leq \rank{\matA} + s, \; \forall \lambda\in\bbC 
\end{align}
\end{theorem}
The connections between \Cref{thm:necc_suff,thm:sparse} are as follows:
\begin{enumerate}[label=(i)]
\item When $\matC=\eye$, the notion of output sparse controllability and sparse controllability are the same. If we substitute $\matC=\eye$ in \Cref{thm:necc_suff}, the rank conditions of \eqref{eq:ness} and \eqref{eq:suff} are equivalent to $\rank{\matW} = N$.
Also, using the arguments presented in the proof of \Cref{cor:necc_suff_2} (see \eqref{eq:cor_inter_4} in \Cref{app:necc_suff_2}),  for all $0\leq i\leq N$
\begin{align}
R_i  &\leq N-\rank{\matA} \\
&= \change{\rank{\matW} - \rank{\matA\matW}} = R_0,
\end{align}
which follows when $\rank{\matW} = N$. Hence, by  \Cref{cor:necc_suff_1}, the system is output $s$-sparse controllable if and only if  \eqref{eq:necc_condition} holds which is equivalent to
\begin{equation}
\rank{\matW} = N \text{ and } s \geq N-\rank{\matA}.
\end{equation}
However, the condition $\rank{\matW} = N$ is equivalent to the rank condition in \eqref{eq:sparse_con} due to the equivalence of the PBH test~\cite{hautus1970stabilization} and Kalman rank test for controllability~\cite{kalman1959general}. In other words, when $\matC=\eye$, \Cref{thm:necc_suff} reduces to \Cref{thm:sparse}. 

\item The proof of \Cref{thm:sparse} given in \cite{joseph2020controllability} is based on the PBH test for controllability whereas our proof of \Cref{thm:necc_suff} is based on the fundamental results in linear algebra. Therefore, the proof of \Cref{thm:necc_suff} provides an alternate method to establish \Cref{thm:sparse}.

\item Comparing \Cref{cor:necc_suff_2} and \Cref{thm:sparse}, we conclude that when the sparsity $s\geq  N-\rank{\matA}$, the system is output $s$-sparse controllable if it is output controllable (i.e., $\matC\matW$ is full row rank); and the system is $s$-sparse controllable if it is controllable (i.e., $\matW$ is full row rank).
\end{enumerate}
\subsubsection{Necessary conditions for output sparse controllability} 
We next present a known set of necessary conditions for output $s$-sparse controllability.
\begin{theorem}[{\cite[Corollary 1]{joseph2020controllability}}]\label{thm:sparse_necc}
Consider the linear dynamical system $\lb\matA,\matB,\matC\rb$ defined in  \eqref{eq:sys}  whose controllability matrix $\matW$ is given by \eqref{eq:W_defn}. The system is output controllable using $s$-sparse vectors only if the following conditions hold:
\begin{align}
\rank{\matC\begin{bmatrix}
\lambda\eye-\matA & \matB
\end{bmatrix}} &= n,  \; \forall \lambda\in\bbC\label{eq:con1}\\
\rank{\matC\matA} &\geq n-s.\label{eq:con2}
\end{align}
\end{theorem}

Our necessary conditions in \Cref{thm:necc_suff} are stronger than those in \Cref{thm:sparse_necc}. To verify this, suppose that \eqref{eq:con1} does not hold, i.e.,  there exist $\lambda\in\bbC$ and $\vecz\in\bbR^n$ such that $\vecz\tran\matC\matA=\lambda\vecz\tran\matC$ and $\vecz\tran\matC\matB=\zero$. In this case, we obtain $\vecz\tran\matC\matW=\zero$ which implies that the rank condition of \eqref{eq:ness} does not hold. Thus, \eqref{eq:con1} is necessary for \eqref{eq:ness} to hold. Also,  the necessary condition \eqref{eq:ness} of \Cref{thm:necc_suff} implies that if the system in  \eqref{eq:sys}  is output $s$-sparse controllable, 
\begin{equation}
R_0 = \change{n-\rank{\matC\matA\matW} }\geq n-\rank{\matC\matA}. 
\end{equation}
As a consequence, \eqref{eq:con2} is necessary for the sparsity bound in \eqref{eq:ness} to hold. Hence, we conclude that \Cref{thm:necc_suff} is stronger than  \Cref{thm:sparse_necc}.
\subsection{Design of Sparse  Control Inputs}\label{sec:design}
\Cref{thm:necc_suff} focuses on the existence of a set of control inputs that ensures output controllability while satisfying the sparsity constraints.  However, another problem related to output controllability is the design of this set of sparse vectors. The problem can be cast as a sparse signal recovery problem using \eqref{eq:sys_concat} where we solve for the unknown sparse vectors $\lc\vecu_k\rc_{k=1}^K$~\cite{sefati2015linear,kafashan2016relating}. 

We first note from \cite[Corollary 2]{joseph2020controllability} that to drive the system from any given initial state $\vecx_0\in\bbR^N$ to any final output $\vecy_f\in\bbR^n$, we need at most $n$ control inputs ($K=n$). Thus, from \eqref{eq:sys_concat}, the design of control inputs reduces to solving for $\tilde{\vecu}=\begin{bmatrix}
\vecu_1\tran & \vecu_2\tran & \ldots & \vecu_n\tran
\end{bmatrix}\tran\in\bbR^{nm}$ using
\begin{equation}\label{eq:piecewise}
\vecy_f -  \matC\matA^n\vecx_0 = \matC\begin{bmatrix}
\matA^{n-1}\matB & \matA^{n-2}\matB &\ldots & \matB
\end{bmatrix}\tilde{\vecu}.
\end{equation} 
Here, the unknown vector $\tilde{\vecu}$ is formed by concatenating $n$ vectors which are $s$-sparse. This signal structure is known as piece-wise sparsity. Hence, \eqref{eq:piecewise} can be efficiently solved (in polynomial time) using  piece-wise sparse recovery algorithms such as the piece-wise orthogonal matching pursuit~\cite{li2016piecewise,sriram2020control}. 

\section{Conclusion}
We derived a set of necessary and sufficient conditions under which a discrete-time linear system is output sparse controllable. Our results apply to any general linear system and do not impose any restrictions on the system matrices. Both necessary and sufficient conditions included a rank condition on the output controllability matrix and a lower bound on the sparsity bound. We also derived the conditions under which both sets of conditions became identical, and showed that the results on output controllability (without any constraints) and controllability (with and without sparsity constraints on the inputs) can be derived as a special case of our result. An important direction for future work is to derive the conditions which are jointly necessary and sufficient for output sparse controllability. Studying output sparse controllability under other constraints on the system like bounded energy, nonnegativity, etc., are also avenues for future work. 

\appendices
\crefalias{section}{appendix}

\section{Proof of \Cref{thm:necc_suff}}\label{app:necc_suff}
\change{The key idea of the proof is to use the Kalman decomposition~\cite[Section 6.4]{chen1998linear} to prove the necessity of \eqref{eq:ness} and sufficiency of \eqref{eq:suff}. The proof also relies on the following results from linear algebra.

\begin{lemma}[\cite{djordjevic2010reverse}] \label{lem:prod}
For any matrix $\matA$ and any orthogonal matrix $\matQ$ of compatible dimension, we have $\lb\matQ\matA\rb^{\dagger} = \matA^\dagger\matQ\tran$.
\end{lemma}

\begin{lemma}\label{lem:rank}
Any matrices $\matA$ and $\matW$ of compatible dimensions satisfies
\begin{equation}
\rank{\matA\matW}=\rank{\matA\matW\matW^{\dagger}}.
\end{equation}
\end{lemma}
\begin{proof}
The result follows because
\begin{multline}
\rank{\matA\matW} \geq \rank{\matA\matW\matW^{\dagger}}
\geq \rank{\matA\matW\matW^{\dagger}\matW}\\=\rank{\matA\matW}.
\end{multline}
\end{proof}
\begin{lemma}\label{lem:spanning_set}
For a given nonzero square matrix $\matA\in\bbR^{N\times N}$, let $R=\rank{\matA^N}$. Then, for any given integers $N\leq p\leq q$, there exist real numbers $\{\alpha_i\}_{i=1}^{R}$ such that 
\begin{equation}\label{eq:spanning_set}
\matA^p = \matA^q\sum_{i=1}^{R}\alpha_i\matA^i.
\end{equation}
\end{lemma}
\begin{proof}
See \Cref{app:spanning_set}.
\end{proof}

Next, we prove the desired result using the above lemmas. We first note that the necessity of the rank condition in \eqref{eq:ness} is straightforward from \Cref{thm:output}. Therefore, to prove \eqref{eq:ness}, it is enough to show that when the system is $s$-sparse output controllable, the lower bound in \eqref{eq:ness} holds. Further, when $m\leq  \max_{0\leq i\leq N-1}  R_i$, the lower bound on sparsity in \eqref{eq:suff} reduces to $s=m$. This case is equivalent to the (non-sparse) output controllability. As a result, the sufficiency of the conditions in \eqref{eq:suff} in this case is straightforward from  \Cref{thm:output}. Thus, it is enough to prove the sufficiency of \eqref{eq:suff} for the case when $\max_{0\leq i\leq N-1}  R_i <m$. 

The proof for the necessity of sparsity bound in \eqref{eq:ness} and sufficiency of \eqref{eq:suff} when $\max_{0\leq i\leq N-1}  R_i <m$ is presented next. At a high level, the proof has the following steps:
\begin{enumerate}[label=\Alph*,leftmargin=0.4cm]
\item \label{st:A}We first use the Kalman decomposition to construct two matrices $\tilde{\matC}\in\bbR^{n\times r }$ and $\tilde{\matA}\in\bbR^{r\times r}$ with $r\triangleq\rank{\matW}$ such that 
\begin{equation}\label{eq:equivalent_scon}
R_i=\rank{\tilde{\matC}\tilde{\matA}^i}-\rank{\tilde{\matC}\tilde{\matA}^{i+1}},\; i\geq 0.
\end{equation}
\item \label{st:B} Using \eqref{eq:equivalent_scon}, we show that when the system is  output $s$-sparse controllable, \eqref{eq:ness} of \Cref{thm:necc_suff} holds. This proof implies that \eqref{eq:ness} is necessary for output $s$-sparse controllability.
\item \label{st:C} In this step, we consider the case when  $\max_{0\leq i\leq N-1}  R_i <m$. To prove the sufficiency of \eqref{eq:suff} in this case,  we assume that the rank condition and sparsity bound in \eqref{eq:suff} hold, i.e.,
\begin{equation}\label{eq:suff_assum}
s\geq \max_{0\leq i\leq N-1}  R_i.
\end{equation}
We then prove that any vector $\vecy\in\bbR^n$, there exists an $s$-sparse vector $\vecu\in\bbR^m$ satisfying
\begin{equation}
\ls \tilde{\matC}_{(i)}-\tilde{\matC}_{(i+1)}\rs\vecy =\ls\eye -
\tilde{\matC}_{(i+1)}\rs\matC\matA^i\matB\vecu,\label{eq:splitz}
\end{equation}
where we define $\tilde{\matC}_{(i)}\in\bbR^{n \times n}$ as
\begin{equation}\label{eq:tildeC_defn}
\tilde{\matC}_{(i)} = \tilde{\matC}\tilde{\matA}^{i}\lb \tilde{\matC}\tilde{\matA}^{i}\rb^{\dagger}.
\end{equation}
\item \label{st:D}Under the assumptions of Step \ref{st:C}, we prove that for any vector $\vecy\in\bbR^n$, there exist $s$-sparse vectors $\lc\vecu_k\in\bbR^m\rc_{k=1}^r$ such that
\begin{equation}\label{eq:zeroeigen}
\vecy =  \sum_{k=1}^{r} \matC\matA^{k-1}\matB\vecu_k+ \tilde{\matC}_{(r)}\ls  \vecy - \sum_{k=1}^{r}\matC\matA^{k-1}\matB\vecu_k\rs.
\end{equation}
\item \label{st:E} Finally, under the assumptions of Step \ref{st:C}, we also show that there exist an integer $0<K<\infty$ and $s$-sparse vectors $\lc\vecu_k\in\bbR^m\rc_{k=r+1}^{K}$ such that
\begin{equation}\label{eq:nonzeroeigen}
\tilde{\matC}_{(r)}\ls  \vecy - \sum_{k=1}^{r}\matC\matA^{k-1}\matB\vecu_k\rs =  \sum_{k=r+1}^{K} \matC\matA^{k-1}\matB\vecu_{k}.
\end{equation}
Combining Steps~\ref{st:D} and \ref{st:E}, we establish the sufficiency of \eqref{eq:suff} when $ \max_{0\leq i\leq N-1}  R_i < m$.
\end{enumerate}}
In the reminder of this section, we provide the details of each step.
\subsection{An Equivalent Definition of $R_i$ in \eqref{eq:Rdefn}}
By the Kalman decomposition~\cite[Section 6.4]{chen1998linear}, there exists an orthogonal matrix $\matQ$ such that
\begin{align}
\matQ&=\begin{bmatrix}
\tilde{\matQ} \in\bbR^{N\times r}& \matR\in\bbR^{N\times N-r}
\end{bmatrix}\in\bbR^{N\times N}\label{eq:Q_defn}\\
\matW &= \begin{bmatrix}
\tilde{\matQ} & \matR
\end{bmatrix} \begin{bmatrix}
\tilde{\matW}\in\bbR^{r\times Nm}\\
\zero\in\bbR^{N-r\times Nm}
\end{bmatrix}\label{eq:W_decom}\\
\matA  &= \begin{bmatrix}
\tilde{\matQ} & \matR
\end{bmatrix}\begin{bmatrix}
\tilde{\matA}\in\bbR^{r\times r} & \matA_{(1)}\\
\zero \in\bbR^{N-r\times r}& \matA_{(2)}
\end{bmatrix}\!\begin{bmatrix}
\tilde{\matQ} & \matR
\end{bmatrix}^{-1}\label{eq:Adecomp}\\
\matB  &= \begin{bmatrix}
\tilde{\matQ} & \matR
\end{bmatrix}\begin{bmatrix}
\tilde{\matB} \in\bbR^{r\times m}\\
\zero\in\bbR^{N-r\times m}
\end{bmatrix}.\label{eq:Bdecomp}
\end{align}
Then, for any integer $i\geq 0$, it is easy to see that
\begin{align}
\matC\matA^i\matW\matW^\dagger
&= \matC\matA^i\matQ \begin{bmatrix}
\tilde{\matW}\\
\zero
\end{bmatrix}\begin{bmatrix}
\tilde{\matW}^\dagger & \zero
\end{bmatrix}\matQ\tran\label{eq:eqv_inter_1}\\
&= \matC\begin{bmatrix}
\tilde{\matQ} & \matR
\end{bmatrix}\begin{bmatrix}
\tilde{\matA} & \matA_{(1)}\\
\zero & \matA_{(2)}
\end{bmatrix}^i\begin{bmatrix}
\eye &\zero\\
\zero&\zero
\end{bmatrix}\matQ\tran\label{eq:eqv_inter_2}\\
&=\matC
\begin{bmatrix}
\tilde{\matQ}\tilde{\matA}^i &\zero
\end{bmatrix}\matQ\tran
\label{eq:eqv_inter_4},
\end{align}
where to get \eqref{eq:eqv_inter_1}, we use \eqref{eq:W_decom} and \Cref{lem:prod}. Also, \eqref{eq:eqv_inter_2} follows from \eqref{eq:Adecomp} and the fact that $\tilde{\matW}$ is a full row rank matrix. Consequently, we conclude that
\change{\begin{equation}
\rank{\matC\matA^i\matW}=\rank{\matC\matA^i\matW\matW^\dagger} = \rank{\matC\tilde{\matQ}\tilde{\matA}^i},
\end{equation}
where we also use \Cref{lem:rank}.}
Thus, we establish \eqref{eq:equivalent_scon} by defining 
\begin{equation}\label{eq:C_decomp}
\tilde{\matC}\triangleq\matC\tilde{\matQ}\in\bbR^{n\times r},
\end{equation}
and Step \ref{st:A} is completed.
\subsection{Necessity of Sparsity Bound in \eqref{eq:necc_cor_2}}
Using \eqref{eq:Adecomp}, \eqref{eq:Bdecomp} and \eqref{eq:C_decomp}, we rewrite \eqref{eq:combinatorial} as
\begin{equation}\label{eq:C_trimmed}
\CS{\tilde{\matC}\begin{bmatrix}
\tilde{\matA}^{K-1}\tilde{\matB}_{\calS_1} & 
\tilde{\matA}^{K-2}\tilde{\matB}_{\calS_2}&\ldots& 
\tilde{\matB}_{\calS_K}
\end{bmatrix}}=\bbR^n.
\end{equation}
Here, for every integer $0\leq i\leq N-1$, the first $(K-i-1)s$ columns of the matrix in \eqref{eq:C_trimmed} belong to $\CS{\tilde{\matC}\tilde{\matA}^{i+1}}$. As a consequence, the last $(i+1)s$ columns of the matrix span the null space of $\tilde{\matC}\tilde{\matA}^{i+1}$. Since the dimension of the null space of $\tilde{\matC}\tilde{\matA}^{i+1}$ is $n-\rank{\tilde{\matC}\tilde{\matA}^{i+1}}$, we deduce that
\begin{equation}
(i+1)s \geq n-\rank{\tilde{\matC}\tilde{\matA}^{i+1}} = \sum_{j=0}^{i} R_j,
\end{equation}
where we use \eqref{eq:equivalent_scon}. Therefore, we obtain that when \eqref{eq:combinatorial} holds, the bound on $s$ given by \eqref{eq:ness} is satisfied. Thus, Step \ref{st:B} is completed.

\change{
\subsection{Characterizing $\CS{\tilde{\matC}_{(i)}-\tilde{\matC}_{(i+1)}}$}
The Kalman decomposition given by \eqref{eq:W_decom}-\eqref{eq:Bdecomp} ensures that the system defined by $\lb \tilde{\matA},\tilde{\matB}\rb$ is controllable~\cite[Section 6.4]{chen1998linear}. So, for any $\vecz\in\bbR^r$, there exist (non-sparse) vectors $\lc\vecv_k\in\bbR^m\rc_{k=1}^r$ such that
\begin{equation}\label{eq:z_decomp}
\vecz = \sum_{k=1}^r\tilde{\matA}^{k-1}\tilde{\matB}\vecv_k.
\end{equation}
This relation immediately implies that for any $\vecz\in\bbR^r$, there exists $\vecv_1\in\bbR^m$ such that
\begin{equation}\label{eq:proj_A}
\ls\tilde{\matC}_{(i)}-\tilde{\matC}_{(i+1)}\rs\tilde{\matC}\tilde{\matA}^i\vecz =\ls\eye-\tilde{\matC}_{(i+1)}\rs\tilde{\matC}\tilde{\matA}^i \tilde{\matB}\vecv_1,
\end{equation}
where we use the following fact from the definition of  $\tilde{\matC}_{(i)}$ given by \eqref{eq:tildeC_defn}: 
\begin{equation}\label{eq:prod_CA}
\ls\tilde{\matC}_{(i)}-\tilde{\matC}_{(i+1)}\rs \tilde{\matC}\tilde{\matA}^k = \begin{cases}
\ls\eye-\tilde{\matC}_{(i+1)}\rs\tilde{\matC}\tilde{\matA}^i & \text{ if } k=i\\
\zero& \text{ if } k>i.
\end{cases}
\end{equation}
The relation \eqref{eq:proj_A} leads to the following: 
\begin{equation}\label{eq:proj_A1}
 \CS{ \ls\tilde{\matC}_{(i)}-\tilde{\matC}_{(i+1)}\rs\tilde{\matC}\tilde{\matA}^i}\subseteq\CS{\ls\eye-\tilde{\matC}_{(i+1)}\rs\matC\matA^i \matB},
\end{equation}
which follows because $\matC\matA^i \matB=\tilde{\matC}\tilde{\matA}^i \tilde{\matB}$ from \eqref{eq:Adecomp}, \eqref{eq:Bdecomp} and \eqref{eq:C_decomp}.
Further, we have
\begin{align}
\rank{ \tilde{\matC}_{(i)}\!-\!\tilde{\matC}_{(i+1)} } &\geq \rank{ \ls \tilde{\matC}_{(i)}-\tilde{\matC}_{(i+1)}\rs \matC\matA^i}\label{eq:subspace_rel1}\\
&\geq\rank{ \ls \tilde{\matC}_{(i)}-\tilde{\matC}_{(i+1)}\rs \tilde{\matC}_{(i)}}\\
&=\rank{ \tilde{\matC}_{(i)}-\tilde{\matC}_{(i+1)} }.\label{eq:subspace_rel}
\end{align}
Here, the last step is due to the below relation which follows from the definition of  the matrix $\tilde{\matC}_{(i)}$ given by \eqref{eq:tildeC_defn} and symmetry of $\tilde{\matC}_{(i)}$ and $\tilde{\matC}_{(i+1)}$:
\begin{equation}\label{eq:CA_simple}
\ls \tilde{\matC}_{(i)} -\! \tilde{\matC}_{(i+1)}\rs \!\tilde{\matC}_{(i)}\!=\tilde{\matC}_{(i)} - \!\ls \tilde{\matC}_{(i)} \tilde{\matC}_{(i+1)} \rs\tran\!\!= \tilde{\matC}_{(i)}- \tilde{\matC}_{(i+1)}.
\end{equation}
From \eqref{eq:subspace_rel}, we obtain 
\begin{equation}\label{eq:subspaces_rank}
\rank{ \tilde{\matC}_{(i)}-\tilde{\matC}_{(i+1)} } = \rank{\! \ls \tilde{\matC}_{(i)}-\tilde{\matC}_{(i+1)}\rs \matC\matA^i}.
\end{equation}
As a result, we arrive at
\begin{align}
\CS{\tilde{\matC}_{(i)}-\tilde{\matC}_{(i+1)} } &= \CS{ \ls\tilde{\matC}_{(i)}-\tilde{\matC}_{(i+1)}\rs\tilde{\matC}\tilde{\matA}^i}\\&\subseteq\CS{\ls\eye-\tilde{\matC}_{(i+1)}\rs\matC\matA^i \matB},
\label{eq:subspcae_set}
\end{align}
which is due to \eqref{eq:proj_A1}. Further, since $\CS{\tilde{\matC}_{(i)}}\subseteq \CS{\tilde{\matC}_{(i+1)}}$, we have
\begin{equation}\label{eq:rank_tildeC1}
\rank{\tilde{\matC}_{(i)}}=\!\rank{ \tilde{\matC}_{(i+1)}}+\rank{ \ls \!\eye- \tilde{\matC}_{(i+1)}\! \rs\!\tilde{\matC}_{(i)}},
\end{equation}
which follows because $\eye-\tilde{\matC}_{(i+1)}$ is the projection on to the subspace orthogonal to $\CS{\tilde{\matC}_{(i+1)}}$. Using \eqref{eq:CA_simple}, we arrive at
\begin{align}
\rank{  \tilde{\matC}_{(i)}- \tilde{\matC}_{(i+1)} } &=\rank{\tilde{\matC}_{(i)}}-\rank{ \tilde{\matC}_{(i+1)}}\\
&=R_i\leq s,\label{eq:rank_tildeC}
\end{align}
which follows from \eqref{eq:equivalent_scon} and assumption \eqref{eq:suff_assum}. Combining \eqref{eq:rank_tildeC} and \eqref{eq:subspcae_set}, we prove \eqref{eq:splitz} because we need at most $s$ columns of $\ls\eye-\tilde{\matC}_{(i+1)}\rs\matC\matA^i \matB$ to span the column space of $\tilde{\matC}_{(i)}-\tilde{\matC}_{(i+1)}$. 
\subsection{Sparse Representation of  the Null Space of $\tilde{\matC}\tilde{\matA}^{r}$}
We prove a more general result: for any vector $\vecy\in\bbR^n$ and integer $1\leq i \leq r$, there exist $s$-sparse vectors $\lc\tilde{\vecu}_k\in\bbR^m\rc_{k=1}^{i}$ such that
\begin{equation}\label{eq:zeroeigen_rel}
\vecy =  \sum_{k=1}^{i} \matC\matA^{k-1}\matB\tilde{\vecu}_k + \tilde{\matC}_{(i)}\ls  \vecy - \sum_{k=1}^{i}\matC\matA^{k-1}\matB\tilde{\vecu}_k\rs.
\end{equation}

We prove \eqref{eq:zeroeigen_rel} using mathematical induction, and for this, we first verify this result for $i=1$. Using \eqref{eq:splitz} of Step \ref{st:C}, for any given $\vecy\in\bbR^n$, there exists an $s$-sparse vector $\tilde{\vecu}_{1}\in\bbR^m$ such that
\begin{equation}\label{eq:splity_0}
\ls \tilde{\matC}_{(0)}- \tilde{\matC}_{(1)} \rs\vecy =\ls\eye -
 \tilde{\matC}_{(1)}\rs\matC\matB\tilde{\vecu}_{1}.
\end{equation}
However, we observe from \eqref{eq:W_decom} that $\CS{\matW}=\CS{\tilde{\matQ}}$, and this observation combined with the rank condition \eqref{eq:suff} of \Cref{thm:necc_suff} leads to the following:
\begin{equation}
\rank{\tilde{\matC}}=\rank{\matC\tilde{\matQ}}= \rank{\matC\matW} = n.
\end{equation}
As a result, we get 
\begin{equation}\label{eq:rankCW}
\tilde{\matC}_{(0)} 
 = \tilde{\matC}\tilde{\matC}^\dagger =\eye.
\end{equation}
Therefore, \eqref{eq:splity_0} yields that for any $\vecy\in\bbR^n$, there exists an $s$-sparse vector $\tilde{\vecu}_{1}\in\bbR^m$ such that 
\begin{equation}\label{eq:splity_1}
\vecy =  \matC\matB\tilde{\vecu}_{1}+\tilde{\matC}_{(1)}\lb \vecy - \matC\matB\tilde{\vecu}_{1}\rb.
\end{equation}
Consequently, \eqref{eq:zeroeigen_rel} holds for $i=1$. 

By inductive hypothesis, we assume that \eqref{eq:zeroeigen_rel} holds for some integer $0\leq i<r$. 
However, in \eqref{eq:zeroeigen_rel}, we have
\begin{equation}
\tilde{\matC}_{(i)}\ls  \vecy - \sum_{k=1}^{i}\matC\matA^{k-1}\matB\tilde{\vecu}_k\rs\in\bbR^n.
\end{equation}
So we again apply \eqref{eq:splitz} to deduce that there exists an $s$-sparse vector $\tilde{\vecu}_{i+1}\in\bbR^m$ such that
\begin{multline}\label{eq:inter_proof1}
\ls \tilde{\matC}_{(i)}- \tilde{\matC}_{(i+1)} \rs \tilde{\matC}_{(i)} \ls  \vecy - \sum_{k=1}^{i}\matC\matA^{k-1}\matB\tilde{\vecu}_k\rs \\
= \ls\eye -\matC_{(i+1)}\rs\matC\matA^i\matB\vecu_{i+1}.
\end{multline}
Combining \eqref{eq:inter_proof1} and \eqref{eq:CA_simple}, we deduce that
\begin{multline}\label{eq:inter_proof}
\tilde{\matC}_{(i)}\ls  \vecy - \sum_{k=1}^{i}\matC\matA^{k-1}\matB\tilde{\vecu}_k\rs \\ = \matC\matA^i\matB\vecu_{i+1} +
 \tilde{\matC}_{(i+1)}\ls \vecy - \sum_{k=1}^{i+1}\matC\matA^{k-1}\matB\tilde{\vecu}_k\rs.
\end{multline}
Adding \eqref{eq:inter_proof} and  the inductive hypothesis \eqref{eq:zeroeigen_rel}, we get
\begin{equation}
 \vecy 
=  \sum_{k=1}^{i+1}\matC\matA^{k-1}\matB\tilde{\vecu}_k 
+ \tilde{\matC}_{(i+1)}\ls \vecy - \sum_{k=1}^{i+1}\matC\matA^{k-1}\matB\tilde{\vecu}_k\rs.
\end{equation}

In conclusion, we obtain that the desired result \eqref{eq:zeroeigen_rel} holds for $i+1$, and thus, the relation \eqref{eq:zeroeigen_rel} is proved. Finally, choosing $i=r$ in \eqref{eq:zeroeigen_rel}, we complete Step~\ref{st:D}.}

\subsection{Sparse Representation of  the Column Space of $\tilde{\matC}\tilde{\matA}^{r}$}
Using \eqref{eq:z_decomp}, for any $\vecz\in\bbR^{r}$, there exist (non-sparse) vectors $\lc\vecv_k\in\bbR^m\rc_{k=1}^r$ such that
\begin{equation}
 \tilde{\matC}\tilde{\matA}^{r}\vecz
 =  \tilde{\matC}\sum_{k=1}^r\tilde{\matA}^{r+k-1}\tilde{\matB}\vecv_k.
\end{equation}
Here, $\vecv_k\in\bbR^m$ can be represented as $
\vecv_k = \sum_{j=1}^{\lceil m/s\rceil}\vecu_{j}^{(k)},$ where $\lc \vecu_j^{(k)}\in\bbR^m\rc_{j,k}$ are all $s$-sparse vectors. Therefore,
 \begin{equation}\label{eq:Aexpand1}
 \tilde{\matC}\tilde{\matA}^{r}\vecz = \tilde{\matC}\sum_{k=1}^r\sum_{j=1}^{\lceil m/s\rceil}\tilde{\matA}^{r+k-1}\tilde{\matB}\vecu_{j}^{(k)}.
 \end{equation}
However, from \Cref{lem:spanning_set},  there exist $\lc\alpha_i^{(j,k)}\in\bbR\rc_{i=1}^{\rank{\tilde{\matA}^r}}$ such that
\begin{equation}\label{eq:Aexpand}
\tilde{\matA}^{r+k-1}\tilde{\matB}\vecu_{j}^{(k)} = \tilde{\matA}^{r+q_{k,j}}\sum_{i=1}^{\rank{\tilde{\matA}^r}}\alpha_i^{(j,k)}\tilde{\matA}^{i}\tilde{\matB}\vecu_{j}^{(k)},
\end{equation}
where we define the quantity $q_{k,j}\geq k-1$ as
\begin{equation}
q_{k,j}\triangleq \ls (k-1)\left\lceil \frac{m}{s}\right\rceil +(j-1)\rs\rank{\tilde{\matA}^r}.
\end{equation}
Substituting \eqref{eq:Aexpand} into \eqref{eq:Aexpand1}, we deduce that 
\begin{align}
 \tilde{\matC}\tilde{\matA}^{r}\vecz
 = \tilde{\matC}\tilde{\matA}^{r}\sum_{k=1}^r\sum_{j=1}^{\lceil m/s\rceil}\sum_{i=1}^{\rank{\tilde{\matA}^r}}\tilde{\matA}^{q_{k,j}+i}\tilde{\matB}\lb \alpha_i^{(j,k)}\vecu_{j}^{(k)}\rb\\
  = \matC\matA^r\sum_{k=1}^r\sum_{j=1}^{\lceil m/s\rceil}\sum_{i=1}^{\rank{\tilde{\matA}^r}}\matA^{q_{k,j}+i}\matB\lb \alpha_i^{(j,k)}\vecu_{j}^{(k)}\rb,\label{eq:zrel}
\end{align}
 which follows from \eqref{eq:Adecomp}, \eqref{eq:Bdecomp}, and \eqref{eq:C_decomp}.
Here, the powers of $\matA$ in each term of the summation are distinct, and $\alpha_i^{(j,k)}\vecu_{j}^{(k)}$ is $s$-sparse, for all values of $i,j$ and $k$. Consequently, for any vector $\vecz\in\bbR^r$, there exists an integer $0<K= r+r\lceil m/s\rceil\rank{\tilde{\matA}^r}<\infty$ and $s$-sparse vectors $\lc\vecu_{k}\in\bbR^m\rc_{k=r+1}^K$ such that
\begin{equation}\label{eq:nonzeroeigen_1}
\tilde{\matC}\tilde{\matA}^{r}\vecz =  \sum_{k=r+1}^{K} \matC\matA^{k-1}\matB\vecu_{k}.
\end{equation}

\change{Finally, we choose 
\begin{equation}
\vecz= \lb \tilde{\matC}\tilde{\matA}^{r}\rb^{\dagger}\ls\vecy - \sum_{k=1}^{r}\matC\matA^{k-1}\matB\vecu_k\rs\in\bbR^r,
\end{equation} in \eqref{eq:nonzeroeigen_1} to complete Step \ref{st:E}, and \Cref{thm:necc_suff} is proved.}

\hfill\qed

\section{Proof of \Cref{lem:spanning_set}}\label{app:spanning_set}
\change{To prove the result, we consider the real Jordan canonical form~\cite{horn2012matrix} of $\matA$:
\begin{equation}
\matA = \matP^{-1}\begin{bmatrix}
\matJ&\zero\\
\zero & \matN
\end{bmatrix}\matP,
\end{equation}
where $\matP\in\bbR^{N\times N}$ is an invertible matrix. Also, $\matJ\in\bbR^{P\times P}$ and $\matN\in\bbR^{N-P\times N-P}$ are formed by the Jordan blocks of $\matA$ corresponding to the nonzero and zero eigenvalues of $\matA$, respectively. In other words, $\matJ$ is an invertible matrix and $\matN$ is a nilpotent matrix, i.e., $\matN^N=\zero$. Consequently, the desired result \eqref{eq:spanning_set} is equivalent to 
\begin{align}
\matP^{-1}\begin{bmatrix}
\matJ^p&\zero\\
\zero & \zero
\end{bmatrix}\matP \notag\\
&\hspace{-1.5cm}= \matP^{-1}\begin{bmatrix}
\matJ^q&\zero\\
\zero & \zero
\end{bmatrix} \lb \sum_{i=1}^{R}\alpha_i\begin{bmatrix}
\matJ^i&\zero\\
\zero & \matN^i
\end{bmatrix}\matP\rb\\
&\hspace{-1.5cm}= \matP^{-1}\begin{bmatrix}
\matJ^q \sum_{i=1}^{R}\alpha_i \matJ^i&\zero\\
\zero & \zero
\end{bmatrix}\matP,
\end{align}
which follows because $p,q\geq N$. Hence, to prove \Cref{lem:spanning_set}, it suffices to show that for any $N\leq p\leq q$, there exist real numbers $\{\alpha_i\}_{i=1}^{R}$ such that
\begin{equation}\label{eq:expand_J}
\matJ^{p-q} = \sum_{i=1}^{R}\alpha_i \matJ^i.
\end{equation}
For this, we first note that 
\begin{equation}\label{eq:Rr}
R=\rank{\matA^N}=\rank{\matJ^N}=\rank{\matJ}=P,
\end{equation}
which is due to the invertibility of $\matJ$. Consequently, by Cayley-Hamilton theorem, we know that the characteristic polynomial of $\matJ$ has degree at most $R$. Hence, for any integer $p-q$, the relation \eqref{eq:expand_J} holds. Thus, the proof is complete.
\hfill\qed}

\section{Proof of \Cref{cor:necc_suff_1}}\label{app:necc_suff_1}
When $\max_{0\leq i\leq N-1} R_i=R_0$, the relation \eqref{eq:suff} reduces to the following:
\begin{equation}\label{eq:bnd_s}
\rank{\matC\matW}=n\text{ and } s\geq R_0=\change{n- \rank{\matC\matA\matW}}.
\end{equation}
Since \eqref{eq:ness} also implies that $s\geq R_0$, and the necessary conditions are less stringent than the sufficiency conditions, we conclude that both \eqref{eq:ness} and \eqref{eq:suff} reduce to \eqref{eq:necc_condition}. So the sufficiency of \eqref{eq:necc_condition} \change{(``if'' part of \Cref{cor:necc_suff_1})} is proved.

Next, we prove the necessity of \eqref{eq:necc_condition} \change{(``only if'' part of \Cref{cor:necc_suff_1})}. For this, suppose that there exists an integer $1\leq i^*\leq N$ such that
$R_0<\max_{0\leq i\leq N-1}  R_i$. Then, 
\begin{align}
\max_{0\leq i\leq N-1}  R_i &>\max_{0\leq i\leq N-1}\frac{R_0+\sum_{j=1}^ i R_{i}}{i+1}\\
&= \max_{0\leq i\leq N-1}\frac{\sum_{j=0}^ i R_{i}}{i+1}.
\end{align}
Thus, the necessary conditions and the sufficient conditions are different when $\max_{0\leq i\leq N-1} R_i\neq R_0$, and our proof is complete.

\hfill\qed
\section{Proof of \Cref{cor:necc_suff_2}}\label{app:necc_suff_2}
From Step \ref{st:A} of the proof of \Cref{thm:necc_suff} given in \Cref{app:necc_suff}, 
\begin{equation}
R_i = \rank{\tilde{\matC}\tilde{\matA}^{i}} -\rank{\tilde{\matC}\tilde{\matA}^{i+1}},
\end{equation}
where $\tilde{\matA}\in\bbR^{r\times r}$ and $\tilde{\matC}\in\bbR^{N\times r}$ are as defined in \eqref{eq:Adecomp} and \eqref{eq:C_decomp}, respectively, and $r=\rank{\matW}$. Using the Sylvester rank inequality~\cite[Section 0.4.5]{horn2012matrix}, we deduce that
\begin{equation}\label{eq:cor_inter}
R_i \leq r-\rank{\tilde{\matA}}=\rank{\matW}-\rank{\tilde{\matA}}.
\end{equation}
Here, we simplify the second term as follows:
\begin{align}
\rank{\tilde{\matA}}&= \rank{\begin{bmatrix}
\tilde{\matA} & \matA_{(1)}\\
\zero & \matA_{(2)}
\end{bmatrix}\begin{bmatrix}
\eye & \zero\\
\zero & \zero
\end{bmatrix}}\notag\\
&= \rank{\matQ\begin{bmatrix}
\tilde{\matA} & \matA_{(1)}\\
\zero & \matA_{(2)}
\end{bmatrix}\matQ^{-1}\matQ\begin{bmatrix}
\eye & \zero\\
\zero & \zero
\end{bmatrix}\matQ^{-1}}\notag\\
&=\rank{\matA\matW\matW^\dagger}=\rank{\matA\matW},\label{eq:cor_inter_31}
\end{align}
where $\matA_{(1)}$ and $\matA_{(2)}$ are defined in \eqref{eq:Adecomp}, and $\matQ$ is defined in \eqref{eq:Q_defn}. Also, \eqref{eq:cor_inter_31} follows from the arguments similar to those in \eqref{eq:eqv_inter_1}-\eqref{eq:eqv_inter_4}, and \Cref{lem:rank}. Substituting  \eqref{eq:cor_inter_31} into \eqref{eq:cor_inter}, we obtain
\begin{equation}
R_i \leq \rank{\matW}-\rank{\matA\matW}\leq N - \rank{\matA},\label{eq:cor_inter_4}
\end{equation}
where we use the Sylvester rank inequality~\cite[Section 0.4.5]{horn2012matrix}.

\change{Hence, using the condition in \Cref{cor:necc_suff_2}, we arrive at
 \begin{equation}
 s\geq  \max_{0\leq i\leq N-1}  R_i \geq \min\lc m, \max_{0\leq i\leq N-1}  R_i\rc.
 \end{equation}
This relation implies that the sufficient condition  \eqref{eq:suff}  of \Cref{thm:necc_suff} holds, and} the desired result follows.

\hfill\qed
\bibliographystyle{IEEEtran}
\bibliography{PControl_cite}

\end{document}